\documentclass[preprint,12pt]{elsarticle}

\usepackage{graphicx}
\usepackage{amssymb}
\usepackage{amsthm}
\usepackage{amsmath}
\newtheorem{theorem}{Theorem}
\theoremstyle{definition}
\newtheorem{definition}{Definition}
\theoremstyle{remark}
\newtheorem{note}{Note}
\newdefinition{remark}{Remark}
\newtheorem{corollary}{Corollary}

\journal{ArXiv}

\begin{document}

\begin{frontmatter}



\title{Polynomial-time Approximation Algorithm for finding Highly Comfortable Team in any given  Social Network}


\author[NITT]{Lakshmi Prabha S}
\ead{jaislp111@gmail.com}

\author[NITT]{T.N.Janakiraman}
\ead{janaki@nitt.edu}


\address[NITT]{Department of Mathematics, National Institute of Technology, 
Trichy-620015, Tamil Nadu, India.}

\begin{abstract}
There are many indexes (measures or metrics) in Social Network Analysis (SNA), like density, cohesion, etc. In this paper, we define a new SNA index called \lq \lq comfortability \rq\rq. One among the lack of many factors, which affect the effectiveness of a group, is \lq \lq comfortability \rq\rq. So, comfortability is one of the important attributes  (characteristics) for a successful team work. It is important  to find a comfortable and successful team in any given social network.   In this paper, comfortable team, better comfortable team and highly comfortable team  of a  social network are defined based on \textbf{graph theoretic concepts} and some of their structural properties are analyzed. 

It is proved that forming better comfortable team or  highly comfortable team in any connected network are NP-Complete using the concepts of domination in graph theory. Next, we give a polynomial-time approximation algorithm for finding such a highly comfortable team in any given network with  performance ratio $O(\ln \Delta)$, where $\Delta$ is the maximum degree of a given network (graph). The time complexity of the algorithm is proved to be $O(n^{3})$, where $n$ is the number of persons (vertices) in the network (graph). It is also proved that our algorithm has reasonably reduced the dispersion rate.

\end{abstract}


\end{frontmatter}


\section{Introduction}
\label{intro}
There are many factors, lack of  which affect the group or team effectiveness. Team processes describe subtle aspects of interaction and patterns of organizing, that transform input into output. 
The team processes will be described in terms of seven characteristics: coordination, communication,
cohesion, decision making, conflict management, social relationships and performance feedback. The readers are directed to refer ~\cite{team} for further details of characteristics of team. In this paper, we discuss about an attribute or characteristic   called \lq\lq COMFORTABILITY \rq\rq, which is also essential for a successful team work. So, we define it as a new index in SNA. Readers are directed to refer~\cite{group} for more details on group dynamics.

Since the beginning of Social Network Analysis, Graph Theory has been a very important tool both to represent social structure and to calculate some indexes, which are useful to understand several aspects of the social context under analysis. Some of the existing indexes (measures or metrics) are betweenness, bridge, centrality, flow betweenness centrality, centralization, closeness, clustering coefficient, cohesion, degree,   density,  eigenvector centrality,  path length. Readers are directed to refer Martino et. al.~\cite{diameter} for more details on indexes in SNA.  In this paper, we define a new index in SNA called \lq comfortability \rq.

Although the definition of comfortability is new, the motivation  for the concept  comes from Wang et al.~\cite{positive}. Wang et. al. have defined \lq \textbf{positive influence}\rq \   based on the graph theoretic concept  \lq \textbf{degree}\rq.  In this paper, we define \lq \textbf{comfortability}\rq \  based on the graph theoretic concept \lq \textbf{eccentricity}\rq, which is based on the metric concept called \lq distance\rq.  

Let the social network be represented in terms of a graph, with the vertex of the graph denotes a person (an actor) in the social network and an edge between two vertices in a graph represents relationship between two persons in the social  network. All the networks  are connected networks in this paper, unless otherwise specified. If the given network is disconnected, then each connected component of the network can be considered and hence it is enough to consider only connected networks.  Hereafter, the word \lq team\rq \ represents induced sub network (sub graph) of a given network (graph). 

Following are some introduction for \textbf{basic graph theoretic  concepts}.
Some basic definitions from Slater et al.~\cite{Slater} are given below.

The  graphs  considered  in  this  paper  are  finite, simple, connected  and  undirected, unless otherwise specified. For a graph $G$, let $V(G)$ (or simply $V$) and $E(G)$ denote its vertex (node) set and edge set respectively and $n$ and $m$ denote the cardinality of those sets respectively. The \textit{degree} of a vertex $v$ in a graph $G$ is denoted by $deg_{G}(v)$. The \textit{maximum degree} of the graph $G$ is denoted by $\Delta(G)$. The length of any shortest path between any two vertices $u$ and $v$ of a connected graph $G$ is called the \textit{distance} between $u$ and $v$ and is denoted by $d_{G}(u,v)$. For a connected graph  $G$, the \textit{eccentricity} $e_{G}(v) = \max\{d_{G}(u,v): u\in V(G)\}$. If there is no confusion, we simply use the notions  $deg(v)$, $d(u,v)$ and $e(v)$ to denote degree, distance and eccentricity respectively for the concerned graph. The minimum and maximum eccentricities are the \textit{radius} and  \textit{diameter} of $G$, denoted by  $r(G)$ and $diam(G)$ respectively. A vertex with  eccentricity $r(G)$ is called a \textit{central vertex} and a vertex with eccentricity $diam(G)$ is called a \textit{peripheral vertex}. A graph $G$ is said to be 
\begin{itemize}
	\item \textit{self-centered}, if $r(G) = diam(G)$;
\item \textit{bi-eccentric}, if $r(G) = diam(G) -1$;
\item \textit{tri-eccentric}, if $r(G) = diam(G) -2$;
\item in general, $(a+1)$- eccentric, if $r(G) = diam(G) - a$.
\end{itemize}
   
   For  $v \in V(G)$, \textit{neighbors} of $v$ are the vertices adjacent to $v$ in $G$. The neighborhood $N_{G}(v)$ of $v$ is the set of all neighbors of  $v$ in $G$. It is also denoted by $N_{1}(v)$. $N_{j}(v)$ is the set of all vertices at distance $j$ from $v$ in $G$. A vertex $u$ is said to be an \textit{eccentric vertex} of $v$, when $d(u, v) = e(v)$. If $A$ and $B$ are not necessarily disjoint sets of vertices, we define
the distance from $A$ to $B$ as
$dist(A,B) = \ min\{ d(a, b) : a \in A, b \in B \}$. \textit{Cardinality} of a set $D$ represents the number of vertices in the set $D$. Cardinality of $D$ is denoted by $|D|$.

 A vertex of degree one is called a \textit{pendant vertex}. A \textit{walk} of length $j$ is an alternating sequence $W:  u_0, e_1, u_1, e_2, u_2,\ldots, u_{j-1}, e_j, u_j$ of vertices and edges with $e_i = u_{i-1}u_i$. If all $j$ edges are distinct, then $W$ is called a \textit{trail}. A walk with $j+1$ distinct vertices $u_0, u_1, \ldots, u_j$ is a \textit{path} and if $u_0 = u_j$ but $u_1, u_2, \ldots, u_j$ are distinct, then the trail is a \textit{cycle}. A path of length $n$ is denoted by $P_n$ and a cycle of length $n$ is denoted by $C_n$. A graph $G$ is said to be \textit{connected} if there is a path joining each pair of nodes. A component of a graph is a maximal connected sub graph. If a graph has only one component, then it is connected, otherwise it is \textit{disconnected}. A \textit{tree} is a connected graph with no cycles (acyclic).

  We say that   $H$ is a \textit{sub graph} of a graph $G$, denoted  by $H < G$, if  $V(H) \subseteq V(G)$  and $uv \in E(H)$ implies $uv \in E(G)$ . If a sub graph $H$ satisfies  the added property that for every pair $u,v$ of vertices, $uv \in E(H)$ if and only if $uv \in E(G)$, then $H$ is called an \textit{induced sub graph} of $G$. The induced sub graph $H$ of $G$ with $S = V(H)$ is called the sub graph induced by $S$ and is denoted by $\left\langle S|G\right\rangle$ or simply $\left\langle S\right\rangle$.
  
  Let $k$ be a positive integer. The $k^{th}$  \textit{power} $G^{k}$ of a graph $G$  has $V(G^{k}) = V(G)$ with $u,v$ adjacent in $G^{k}$ whenever $d(u,v) \leq k$. A graph is said to be \textit{complete} if each vertex in the graph  is adjacent to every other vertex in the graph. A \textit{clique} is a maximal complete sub graph.

  The concept of domination was introduced by Ore~\cite{Ore} . Readers are  directed to refer Slater et al.~\cite{Slater}.   A set $D \subseteq V(G)$  is  called a \textit{dominating set} if every vertex  $v$ in $V$ is either an element of $D$ or is adjacent to an element of  $D$.  A  dominating set  $D$ is a minimal dominating set if  $D-\{v\}$  is  not a dominating set for any $v \in D$. The domination number  $\gamma(G)$ of a graph $G$ equals the minimum cardinality of a dominating set in $G$. 
  
  A set $D$ of vertices in a connected graph $G$ is called a \textit{$k$-dominating set} if every vertex in $V- D$ is within distance $k$ from some vertex of $D$. The concept of the $k$-dominating set was
introduced by Chang and Nemhauser~\cite{Chang,Chang2} and could find applications for many situations and structures which give rise to graphs; see the books by Slater et al~\cite{Slater, REF6}. So, dominating set is nothing but 1-distance dominating set.
  
  Sampath  Kumar  and  Walikar~\cite{CDS}  defined  a connected  dominating  set   $D$  to  be  a dominating set $D$,  whose  induced   sub-graph  $\left\langle D\right\rangle$ is \textit{connected}.  The  minimum  cardinality  of  a connected  dominating  set  is  the  connected domination number $\gamma_c(G)$. 
  
The readers are also directed to refer Slater et al.~\cite{Slater} for  further details of basic definitions, not given in this paper.

The notation floor($x$) = $ \left\lfloor x \right\rfloor$ is the largest integer not greater than $x$ and ceiling($x$) = $ \left\lceil x\right\rceil $ is the smallest integer not less than $x$. For example, $ \left\lfloor 3.5 \right\rfloor = 3$ and  $ \left\lceil 3.5 \right\rceil = 4$.

Let us recall the terminologies as follows: The symbol ($\rightarrow$)  denotes  \lq\lq represents \rq\rq 
\begin{itemize}
	\item Graph $\rightarrow$ Social Network (connected)
	\item Vertex of a graph $\rightarrow$ Person in a social network
	\item Edge between two vertices of a graph $\rightarrow$ Relationship between two persons in a social network
	\item Induced subgraph of a graph $\rightarrow$  \textbf{Team or Group} of a social network.
\end{itemize}

Given a connected network of people.  Our problem is to find a team (sub graph) which is less dispersive, highly flexible and performing better. 
Let us first discuss the characteristics of a  good performing (successful) team. 
\begin{definition}We define a team  to be  \textbf{good performing or successful} if the team is 
\begin{enumerate}
	\item less dispersive
	\item having good communication among the team members
	\item easily accessible to the non- team members 
	\item a good service provider to the non-team members (for the whole network).
\end{enumerate}
\end{definition}
Next, let us mathematically formulate these four characteristics.\\
For any given network, the comfortable team should be dominating. A team $D$ is said to be dominating if   at least one person in the team is accessible to every person not in the team. Domination is an important criteria for any network. So, the persons in the team should first of all be dominating the entire network. Domination represents the fourth characteristic, that is if a team is dominating, it means that the team  is a good service provider to the non-team members. \\
\textbf{Domination $\rightarrow$ good service provider to the non-team members}. 

There should always be some communication between the persons in the network. The dominating team should be connected so that they discuss among themselves and as a team act for the welfare of the whole network. Connectedness represents the second characteristic.\\
\textbf{Connectedness $\rightarrow$ good communication among team members}.

So, any team should always be dominating and connected. The other two characteristics will be mathematically formulated in the later sections. When is a team or set called less dispersive? Let us discuss in the coming section.
\begin{note}
\textbf{Notation 1:}\\
In all the figures of this paper, 
\begin{itemize}
	\item $\{v_1, v_2,\ldots, v_n\}$ represent the vertex set of the graph $G$, that is,\\ $V(G) = \{v_1, v_2,\ldots, v_n\}$. 
	\item The numbers besides every vertex represents the eccentricity of that vertex. 
For example, in Figure~\ref{fig1}, in the graph $G$, $e(v_1) = 5$, $e(v_2) = 4$, $e(v_3) = 3$, $e(v_4) = 3$, $e(v_5) = 4$ and $e(v_6) = 5$. 
\item The set notation  $D = \{v_1, v_2, \ldots, v_n\}$ represents only the individual persons but does not represent the relationship between them.
\item The notation $\left\langle D \right\rangle$ represents the team. $\left\langle D \right\rangle$ is the induced sub graph of $G$, which represents the persons as well as the relationship between them.
So, the set $D$ represents only the team members and the team represents the persons with their relationship.
\end{itemize}
\end{note}
The remaining part of the paper is organized as follows:
\begin{itemize}
	\item Section~\ref{comf} defines comfortable team and analyses the advantages and disadvantages of the comfortable team.
\item Section~\ref{BC} discusses about better comfortable team,  its properties, advantages and disadvantages.
\item Section~\ref{HC} defines highly comfortable team and discusses some of its properties.
\item In section~\ref{Algo}, an approximation algorithm for finding highly comfortable team in any given network is given with illustrations. Time complexity of the algorithm is analyzed. 
\item In Section~\ref{correct}, some theorems are proved which supports the correctness of the algorithm.
\item Section~\ref{rate} discusses about how algorithm helps in reduction of dispersion rate.
\item Section~\ref{ratio} proves some theorems which gives the performance ratio of the algorithm.
\item Section~\ref{max} defines highly comfortable team with maximum members and discusses about the advantages of minimum highly comfortable team and maximum highly comfortable team.
\item Section~\ref{conc} concludes the paper and discusses about some future work.
\end{itemize}
\section{Comfortable Team}
\label{comf}
An important "descriptive" index which is easy to calculate, but gives us important information about the closeness of the vertices in the graph, is the diameter.  Martino et. al. ~\cite{diameter} has given that \lq\lq The longer the diameter is, the more a graph (network) is dispersive\rq\rq. Let us coin it in graph theoretical terms as follows:
If $d(u,v) = diam(G)$, then the vertex (person) $u$ is said to be dispersive from the vertex (person) $v$. That is, the person $u$ is far away from the person $v$ and $u$ feels uncomfortable to pass any information to the person $v$. Similarly, any person in the network is uncomfortable with all the persons in his $d^{th}$ neighborhood, where $d=diam(G)$. Also, any person ($u$) in the network is uncomfortable with all the persons in his farthest set (set of all eccentric vertices of the vertex $u$). Thus, we define the less dispersive set  as follows:
\begin{definition}\textbf{Less Dispersive Set:}
A set $D$ is said to be less dispersive, if $e_{\left\langle D\right\rangle}(v) <  e_G(v)$, for every vertex $v \in D$.
\end{definition}
\begin{definition}\textbf{Less Dispersive Dominating Set:}
A set $D$ is said to be a less dispersive dominating set if the set $D$ is dominating, connected and less dispersive. The cardinality of minimum  less dispersive dominating set of $G$ is denoted by $\gamma_{comf}(G)$. A set of vertices is said to be a  $\left\langle \gamma_{comf}-set \right\rangle$, if it is a less dispersive dominating set with cardinality  $\gamma_{comf}(G)$. 
\end{definition}
\begin{definition}\textbf{Comfortable Team:}
A team $\left\langle D \right\rangle$ is said to be a comfortable team if $\left\langle D \right\rangle$  is  less dispersive and dominating. Minimum comfortable team is a comfortable team with the condition:  $|D|$ is minimum.
\end{definition}
\textbf{Example 1:} Consider the graph (network) $G$ in Figure~\ref{fig1}.  Here, $G$ is a path of length six ($P_6$). $D= \{v_2,v_3,v_4,v_5\}$. The induced sub graph  $\left\langle D \right\rangle$ of $G$ forms a path of length four ($P_4$) and so it dominates all the vertices in $V-D$. Also, $\left\langle D \right\rangle$  forms the comfortable team of $G$,  because \\$e_{\left\langle D\right\rangle}(v_2) = 3 < 4 =   e_G(v_2). \  \Rightarrow e_{\left\langle D\right\rangle}(v_2) <  e_G(v_2)$. Similarly, $e_{\left\langle D\right\rangle}(v_i) <  e_G(v_i)$ for every $i = 3,4,5$. Thus, $D$ forms less dispersive set and  hence $\left\langle D \right\rangle$  forms the comfortable team of $G$. $\Rightarrow \gamma_{comf}(P_6) = 4$.
\begin{figure}
\footnotesize\centering
\centerline{\includegraphics[width=3in]{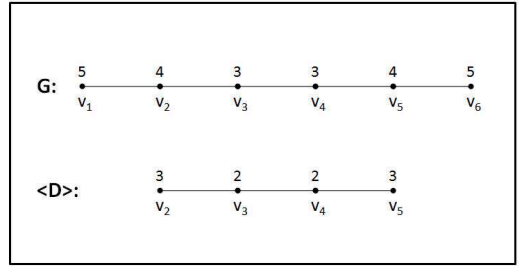}}
\caption{A Network and its Comfortable Team}
\label{fig1}
\end{figure}

So, the problem is coined as: Find a team which is dominating, connected and less dispersive. 
It is  to be noted that there are many graphs  which do not have $\left\langle \gamma_{comf}-set \right\rangle$. So, we must try to avoid such kind of networks for successful team work.\\
\textbf{Example 2:} Consider the graph $G$ in Figure~\ref{fig2}. Here, $G$ is a cycle of length six ($C_6$).
The vertices $v_1$ and $v_4$ dominate all the vertices of $G$. So, with connectedness, we can take $D=\{v_1, v_2, v_3, v_4\}$. The set $D$ dominates $G$, but $D$ is not less dispersive, because,\\
$e_{\left\langle D\right\rangle}(v_1) = 3 =  e_G(v_1)$ and  $e_{\left\langle D\right\rangle}(v_4) = 3 =  e_G(v_4)$. The vertices $v_1$ and $v_4$  maintained the original  eccentricity as in $G$. Thus, $e_{\left\langle D\right\rangle}(v) <  e_G(v)$, for every vertex $V \in D$ is \textbf{not satisfied}. So, $D$ is not less dispersive and hence $\left\langle D\right\rangle$ is not a comfortable team. \\
Also, $D_1 =\{v_1, v_2, v_3\}$ forms less dispersive set in $G$, (from Figure~\ref{fig2}), but $D_1$ is not dominating. The vertex $v_5$ is left undominated. \\
From the above discussion, we get, 
\begin{itemize}
	\item  the less dispersive set may not be dominating 
\item  the dominating set may not be less dispersive.
\end{itemize} 
So, under one of these two cases, the graph $G$ does not possess comfortable team. It is to be noted that not only $C_6$, but all the cycles $C_n$, do not possess a comfortable team. Also, there are infinite families of graphs which do not possess comfortable team.
\begin{figure}
\footnotesize\centering
\centerline{\includegraphics[width=3in]{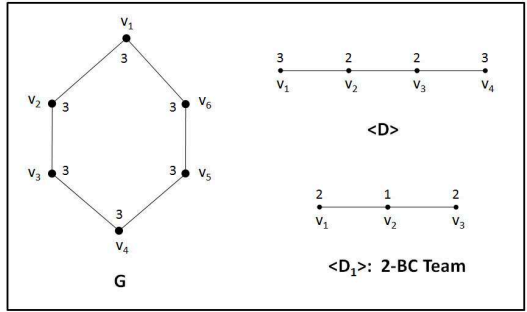}}
\caption{A Network and its 2-BC Team}
\label{fig2}
\end{figure}
\\
\textbf{Disadvantage of the comfortable team:}\\ 
We can see that if a set $D$ is $\left\langle \gamma_{comf}-set \right\rangle$, then the diameter $diam(\left\langle D\right\rangle)$ is reduced by one or two and the radius $r(\left\langle D\right\rangle)$ is reduced by one (by the property of domination). So, under the view of Martino et al.~\cite{diameter},  the team $\left\langle D \right\rangle$ is still dispersive and hence uncomfortable. Also, as discussed in the Example 2, comfortable team does not exist in any given network. Infinite families of networks do not possess comfortable team. As this comfortable team has only minimum applications,  we do not concentrate much on the comfortable team in this paper. 
 
The main aim of this paper is to find a team which is \textbf{more comfortable and less dispersive}, in \textbf{any given network}.\\
So, we define a better comfortable team with one more condition in the next section.
\section{Better Comfortable Team}
\label{BC}
\begin{definition}\textbf{$l$-Better Reduced Dispersive Set:}
A set $D$ is said to be $l$-better reduced dispersive set if 
\begin{enumerate}
	\item $e_{\left\langle D\right\rangle}(v) <  e_G(v)$, for every vertex $v \in D$ (less dispersive) and 
	\item $diam(\left\langle D\right\rangle) \leq \left\lceil \displaystyle \frac{diam(G)}{l}\right\rceil$ (measure of dispersion within $D$, which represents good communication among team members).
\end{enumerate}
\end{definition}
Next, we define the better reduced dispersive set  with respect to domination.
\begin{definition}\textbf{$l$-Better Reduced Dispersive $k^{*}-$ Distance Dominating Set:}
A set $D$ is said to be a \lq $l$-better reduced dispersive\rq \ $k^{*}-$ distance dominating set if $D$ is a $l$-better reduced dispersive set and a $k^{*}-$ distance dominating set. That is, 
\begin{enumerate}
\item $e_{\left\langle D\right\rangle}(v) <  e_G(v)$, for every vertex $v \in D$ (less dispersive)  
	\item $diam(\left\langle D\right\rangle) \leq \left\lceil \displaystyle \frac{diam(G)}{l}\right\rceil$ (measure of dispersion within $D$, which represents good communication among team members) and 
	\item $dist(D, V-D) \leq k^{*}$ ($k^{*}-$ distance domination). ($k^{*}$ is the measure of dispersion between $D$ and $V-D$. It is also called dispersion index).
 
\end{enumerate}
Minimum cardinality of a \lq $l$-better reduced dispersive\rq \ $k^{*}-$ distance dominating set of $G$ is denoted by $\gamma_{lbcomf}(G)$.
\end{definition}

\begin{definition}\textbf{Better Comfortable Team:}
A team $\left\langle D \right\rangle$ is said to be a $l$-Better Comfortable ($l$-BC) team if  $\left\langle D \right\rangle$  is $l$-better reduced dispersive, $k^{*}-$ distance dominating.  $l$-min BC team is a $l$-BC team with the condition: $|D|$ and $k^{*}$ are minimum.
\end{definition}
\textbf{Example 3:} Consider the graph $G$ ($C_6$) in Figure~\ref{fig2}. In $C_6$, $D_1= \{v_1, v_2, v_3\}$ forms a $2$-better reduced dispersive, 2-distance dominating set, because
\begin{enumerate}
	\item  $e_{\left\langle D_1\right\rangle}(v_i) <  e_G(v_i)$, for $i = 1,2,3$.
	\item $diam(\left\langle D_1\right\rangle) = 2$ and  $\left\lceil \displaystyle \frac{diam(G)}{l}\right\rceil = \left\lceil \displaystyle \frac{3}{2} \right\rceil = 2$ and hence \\ $diam(\left\langle D_1\right\rangle) = \left\lceil \displaystyle \frac{diam(G)}{l}\right\rceil$.
	\item $k^{*} = 2$, because $v_5$ is reachable from $D_1$ by  distance  two, $v_4$ and $v_6$ are reachable from $D_1$ by distance one. $\Rightarrow dist(D_1, V-D_1) \leq 2$. 
\end{enumerate}
Thus, for $l = 2$, $l$-BC team exists in $C_6$ and hence $\gamma_{2bcomf}(G) = 3$.
\begin{note}
It is to be noted that in this section, we have defined $k^{*}-$ distance dominating set, not simply dominating set, because, by definition, a $l$-better reduced dispersive set will always have $diam(\left\langle D\right\rangle) \leq \left\lceil \displaystyle \frac{diam(G)}{l}\right\rceil$ and hence  there may not always exist a set $D$ at distance one from the set $V-D$. For example, if $l=2$, then, as discussed in the Example 3, $k^{*} = 2$ for $C_6$ and hence dominating set may not be possible, but $k^{*}$-distance dominating set is possible. 
\end{note}
\begin{theorem}
\label{better NPC}
Forming $l$-better comfortable team in a given network is NP-complete.
\end{theorem}
\begin{proof}
Let $D$ be a minimum $l$-better reduced dispersive $k^{*}-$ distance dominating set of $G$.\\
$\Rightarrow D$  is a connected $k^{*}-$ distance dominating set of $G$ (since any better reduced dispersive set is a connected set).\\
$\Rightarrow D$  is a connected dominating set of $G^{k^{*}}$ (by definition of the graph $G^{k^{*}}$). \\
Finding $\gamma_c(G)$, for any graph $G$ is NP-complete (by Slater et al.~\cite{Slater}).\\
$\Rightarrow$ Finding $\gamma_c(G^{k^{*}})$ is NP-complete.\\
$\Rightarrow$ Finding minimum \lq $l$-better reduced dispersive\rq \  $k^{*}-$ distance dominating set of $G$ is NP-complete (by above points). \\
Thus forming $l$-better comfortable team in a given network is NP-complete.
\end{proof}
\textbf{Disadvantage of Better Comfortable Team}\\
The less dispersive set $D$ is made better reduced dispersive by fixing the diameter of $D$ and hence dispersiveness is reduced comparatively.  $dist(D, V-D) \leq k^{*}$. Now, $l$-BC team satisfies all the characteristics of a good performing team (discussed in Section~\ref{intro}) except the third one, because,  if $k^{*} > diam(\left\langle D \right\rangle)$, then the maximum distance among the vertices in $D$ is lesser than the distance between the vertices in $D$ and $V-D$. It means that the team members are less dispersive but persons not in the team (vertices in the set $V-D$)  have difficulty in accessing the  team members (vertices in $D$). It is not fair to make $D$ comfortable and $V-D$ having uncomfortability to reach $D$. A team is formed in the network only to serve for the whole network. So, in the next section, we define a highly comfortable team, maintaining comfortability inside $D$ and accessibility between $D$ and $V-D$.\\ 
\textbf{Advantages of Better Comfortable Team}\\
The following are some advantages of better comfortable team.
\begin{itemize}
	\item The better comfortable team always exists in any given network.
	\item If $k^{*} \leq diam(\left\langle D \right\rangle)$, then the better comfortable team itself is highly comfortable.
\end{itemize}

\section{Highly Comfortable team}
\label{HC}
\begin{definition}
\label{all}
\textbf{$l$-Highly Reduced Dispersive $k^{*}-$ Distance Dominating Set:}
A set $D$ is said to be a \lq $l$-highly reduced dispersive\rq \ $k^{*}-$ distance dominating set if
\begin{enumerate}
 
\item $D$ is a $l$-better reduced dispersive  $k^{*}-$ distance dominating set and 
\item $k^{*} \leq diam (\left\langle D\right\rangle)$. (easily accessible from the non-team members)
\end{enumerate}
That is, 
\begin{enumerate}
\item $e_{\left\langle D\right\rangle}(v) <  e_G(v)$, for every vertex $v \in D$ (less dispersive)  
	\item $diam(\left\langle D\right\rangle) \leq \left\lceil \displaystyle \frac{diam(G)}{l}\right\rceil$ (measure of dispersion within $D$, that is, good communication among team members)  
	\item $dist(D, V-D) \leq k^{*}$ ($k^{*}-$ distance domination, that is, good service providers to the non team members)
	and \item $k^{*} \leq diam (\left\langle D\right\rangle)$ (easily accessible from the non team members).
 
\end{enumerate}
Minimum cardinality of a  \lq $l$-highly reduced dispersive\rq \ $k^{*}-$ distance dominating set of $G$ is denoted by $\gamma_{lhcomf}(G)$.
\end{definition}
\begin{definition}\textbf{$l$-Highly Comfortable Team:}
A team $\left\langle D \right\rangle$ is said to be $l$-Highly Comfortable ($l$-HC team) if  $\left\langle D \right\rangle$  is $l$-highly reduced dispersive,  $k^{*}-$ distance dominating. $l$-min HC team is a $l$-HC team with the condition:  $|D|$ and  $k^{*}$ are minimum.
\end{definition}
Thus, from the Definition~\ref{all}, it is clear that $l$-HC team satisfies all the characteristics of a good performing team, mentioned in the Definition 1 in  Section~\ref{intro} .
The third characteristic given in the Section~\ref{intro} is mathematically formulated as:\\
If $k^{*} \leq diam (\left\langle D\right\rangle)$, then it means that the team is easily accessible from the non-team members (since $k^{*}$ denotes the distance between $D$ and $V-D$).

\subsection{Properties}
First, we prove the NP-completeness of forming $l$-HC team in a given network.
\begin{theorem}
Forming $l$-HC team in a given network is NP-complete.
\end{theorem}
\begin{proof}
As any $l$-highly reduced dispersive set is a $l$-better reduced dispersive set (a connected set), the proof follows from Theorem~\ref{better NPC}.
\end{proof}
\section{Approximation Algorithm}
\label{Algo}
In this section, we give a polynomial-time approximation algorithm for finding $l$-HC  team from a given network.
\subsection{Notation 2}
\begin{itemize}
	\item $D \rightarrow $ minimum $l$-Highly reduced dispersive, $k^{*}$-distance dominating set.
	\item $D_1 \rightarrow $ output of our algorithm, which is a minimal $l$-Highly reduced dispersive, $k^{*}$-distance dominating set. $\Rightarrow |D_1| \geq |D|$.
	\item $k^{*} \rightarrow$ the distance between two sets $D$ and $V-D$, that is, \\ $dist(D,V-D) \leq k^{*}$.
	\item $k \rightarrow$ the distance between two sets $D_1$ and $V-D_1$, that is, \\ $dist(D_1,V-D_1) \leq k$.
	\item $d_1 =$ upper bound of $diam(\left\langle D\right\rangle)$. So,  $d_1 =  \left\lceil \displaystyle \frac{diam(G)}{l}\right\rceil$.
	\item $diam(\left\langle D_1\right\rangle) \leq \left\lceil \displaystyle \frac{diam(G)}{l}\right\rceil$ at the intermediate stages.\\
	Finally, $ diam(\left\langle D_1\right\rangle) =   \left\lceil \displaystyle \frac{diam(G)}{l}\right\rceil = d_1$.  
	\item Performance ratio = $\displaystyle \frac{|minimal\ set|} {|minimum\ set|}$.
\end{itemize}
\subsection{Algorithm HICOM} 
A polynomial time approximation algorithm for finding $l$-HC team is given below.\\
Input: $G$.\\
Output: $D_1$, which is a $l$-Highly Reduced Dispersive $k^{*}$-distance dominating set, so that $\left\langle D_1\right\rangle$ is a $l$-HC team.\\
First choose any $l$ from the set of positive real numbers. Let $d_1 =$ upper bound of $diam(\left\langle D\right\rangle)$. So,  $d_1 =  \left\lceil \displaystyle \frac{diam(G)}{l}\right\rceil$.\\   
HICOM(G)
\begin{enumerate} 
\item Choose a central vertex $v$ (ties can be broken arbitrarily) and add it to $D_1$.
\item If $d_1$ is even, then choose all the vertices in $N_j(v)$, for $j \leq \displaystyle \frac{d_1}{2}$ and add them to $D_1$.\\ 
else choose all the vertices in $N_j(v)$, for $j \leq \displaystyle \frac{(d_1-1)}{2}$ and add them to $D_1$.
\item Put $i = \left\lfloor \displaystyle \frac{d_1}{2} \right\rfloor$.
	\item If $diam(\left\langle D_1\right\rangle) = \left\lceil \displaystyle \frac{diam(G)}{l}\right\rceil$, then Goto step 7, else Goto next step (step 5).
	\item Put $i=i+1$.
\item Choose a vertex from $N_i(v)$ and add it to $D_1$. Then GOTO step 4.
\item If $e_{\left\langle D_1 \right\rangle}(v) < e_G(v)$, for every vertex $v \in D_1$, then print $D_1$, \\ else suitably remove some vertices from $D_1$ such that $\left\langle D_1 \right\rangle$ maintains the conditions in steps 4 and 7.
\item Stop.	
\end{enumerate} 
\begin{note}
\label{note3}
It is to be noted that at the end of the algorithm, $ diam(\left\langle D_1\right\rangle) =   \left\lceil \displaystyle \frac{diam(G)}{l}\right\rceil = d_1$.  

\end{note} 
\subsection{Illustrations}
\textbf{Example 4:} Consider the graph $G$ as in Figure~\ref{fig3}.  First let us fix $l = 2$. $\Rightarrow d_1 =  \left\lceil \displaystyle \frac{diam(G)}{l}\right\rceil =  \left\lceil \displaystyle \frac{9}{2}\right\rceil = 5 $. 

There are four central vertices in $G$, namely, $v_1, v_6, v_{11}$ and $v_{16}$. We can start the algorithm from any vertex. Let us start from  $v_1$ and add it to $D_1$. 

Here, $d_1$ is odd. So, $j \leq \displaystyle \frac{(d_1-1)}{2} = 2$. So, we take all the vertices from $N_1(v_1)$ and $N_2(v_1)$ and add them to $D_1$. At this stage, $D_1 = \{v_1, v_2, v_{21}, v_{20}, v_3, v_{22}, v_{19}\}$. 

Now,  we can see that $diam(\left\langle D_1\right\rangle) = 4 < 5= \left\lceil \displaystyle \frac{diam(G)}{l}\right\rceil$. So, in order to make $diam(\left\langle D_1\right\rangle) = 5$, we add one more vertex $v_4$ from $N_3(v_1)$ and add it to $D_1$. 

At this stage,  $D_1 = \{v_1, v_2, v_{21}, v_{20}, v_3, v_{22}, v_{19}, v_4\}$ and \\ $diam(\left\langle D_1\right\rangle) = 5= \left\lceil \displaystyle \frac{diam (G)}{l}\right\rceil$. Also, (from the Figure~\ref{fig3}),  we can see that $D_1$ satisfies $e_{\left\langle D_1 \right\rangle}(v) < e_G(v)$, for every vertex $v \in D_1$. 

Also, every vertex in $V-D_1$ is reachable from $D_1$ by a distance lesser than or equal to five, that is, $dist(D_1, V-D_1) \leq 5$ $\Rightarrow  k= 5$ and hence $D_1$ satisfies  the final condition $k = diam (\left\langle D_1\right\rangle)$. 

The output is $D_1 = \{v_1, v_2, v_{21}, v_{20}, v_3, v_{22}, v_{19}, v_4\}$, which is the $2$-highly reduced dispersive, $5$ - distance dominating set and hence  the $2$-HC team is as shown in the Figure~\ref{fig3}.
\begin{figure}
\footnotesize\centering
\centerline{\includegraphics[width=5in]{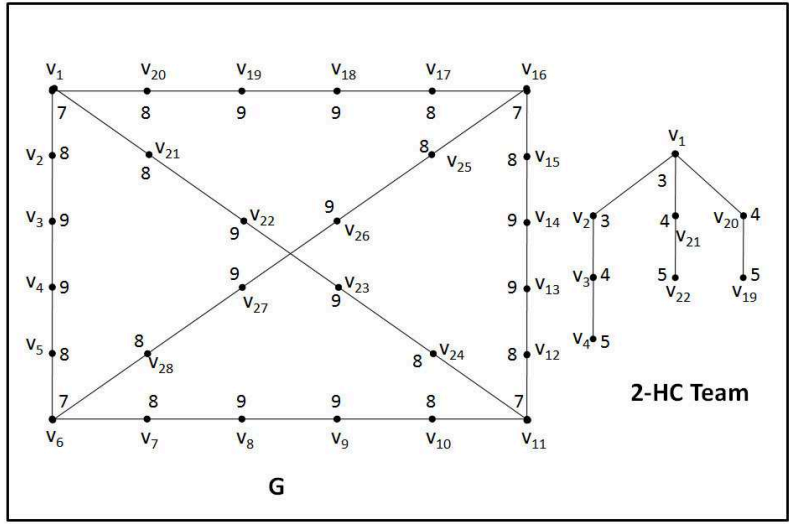}}
\caption{A Network and its 2-HC Team}
\label{fig3}
\end{figure}
\\\textbf{Example 5:} Consider the graph $G$ as in the Figure~\ref{fig4}. Let us fix $l =\displaystyle \frac{3}{2}$. $\Rightarrow d_1 =  \left\lceil \displaystyle \frac{diam(G)}{l}\right\rceil =  \left\lceil \displaystyle \frac{6}{1.5}\right\rceil = 4 $. $\Rightarrow d_1$ is even and hence $j \leq \displaystyle \frac{d_1}{2} = 2$. So, in this case, we arbitrarily start from $v_{11}$ and take all the vertices from $N_1(v_{11})$ and $N_2(v_{11})$ and add them to $D_1$. $D_1 = \{v_{11},v_{12}, v_1, v_{20}, v_{13}, v_2, v_{10}, v_{19}\}$. This set $D_1$ satisfies all the conditions of $\displaystyle \frac{3}{2}$- highly reduced dispersive, 4-distance dominating set and hence the $\displaystyle \frac{3}{2}$-HC team is as shown in the Figure~\ref{fig4}. 
\begin{figure}
\footnotesize\centering
\centerline{\includegraphics[width=4in]{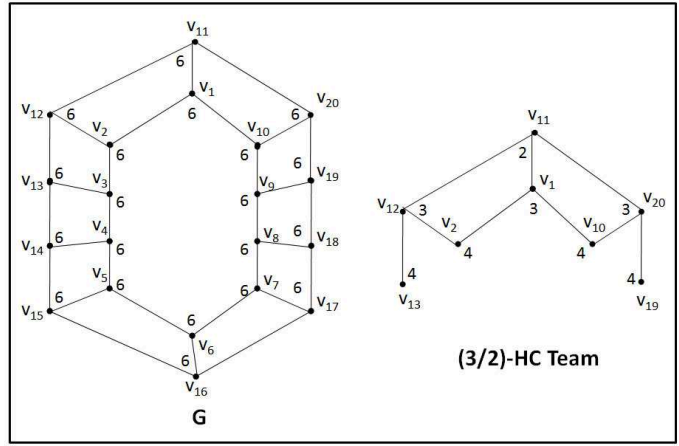}}
\caption{A Network and its $\displaystyle \frac{3}{2}$-HC Team}
\label{fig4}
\end{figure}

\subsection{Time Complexity of the Algorithm HICOM}
Let us discuss the time complexity of the algorithm as follows:
 
The definition of $l$-HC team is dependent on eccentricity of every vertex. So, we have to find eccentricity of every vertex of $G$. By Performing Breadth-First Search (BFS) method from each vertex, one can determine the distance from each vertex to every other vertex. The worst case time complexity of BFS method for one vertex is $O(n^{2})$. As the BFS is method is done for each vertex of $G$, the resulting algorithm has worst case time complexity $O(n^{3})$. As eccentricity of a vertex $v$ is defines as $e(v) = \max\{d(u,v): u\in V(G)\}$, finding eccentricity of vertices of $G$ takes at most $O(n^{3})$.

Thus, the total worst case time complexity of the algorithm is at most $O(n^{3})$.
\section{Correctness of the Algorithm HICOM}
\label{correct}
In order to prove that our algorithm yields a $l$-HC team, it is enough  to prove that $k^{*} \leq diam(\left\langle D \right\rangle)$. As $D_1$ is the output of the algorithm HICOM and $dist(D_1,V-D_1) \leq k$, we have to prove that $k \leq diam(\left\langle D_1 \right\rangle)$. \\ Let 
\begin{equation*}
x = 
\begin{cases}
\displaystyle \frac{d_1}{2}, & \text{if $d_1$ is even}\\ 
\displaystyle \frac{(d_1-1)}{2}, & \text{otherwise.}
\end{cases}
\end{equation*}
From  Note~\ref{note3}, $d_1 = \left\lceil \displaystyle \frac{diam(G)}{l}\right\rceil = diam(\left\langle D_1 \right\rangle)$.
So, $x$ can be written as follows:
\begin{equation*}
x = 
\begin{cases}
\displaystyle \frac{diam(\left\langle D_1 \right\rangle)}{2}, & \text{if $diam(\left\langle D_1 \right\rangle)$ is even}\\ 
\displaystyle \frac{[diam(\left\langle D_1 \right\rangle)-1]}{2}, & \text{otherwise.}
\end{cases}
\end{equation*}
The following theorem helps in proving the correctness of the algorithm HICOM.
\begin{theorem}
$k \leq r(G) -x$.
\end{theorem}
\begin{proof}
From the algorithm, \\
\textbf{Case(1):} Suppose $d_1$ is even, at the end of step 2 (after step 2 is executed and before executing step 3).
\begin{itemize}
	\item At the end of step 2, if  $\left\langle D_1 \right\rangle$ is a tree, then $k = r(G) - x$.
\item At the end of step 2, if  $\left\langle D_1 \right\rangle$  contains cycles and $diam(\left\langle D_1 \right\rangle) \neq \left\lceil \displaystyle \frac{diam(G)}{l}\right\rceil$, then in order to make $diam(\left\langle D_1 \right\rangle) = \left\lceil \displaystyle \frac{diam(G)}{l}\right\rceil$, we insert some more vertices from the next neighborhoods of $v$ in to $D_1$ in step 6 of the algorithm. Hence $k$ may become lesser than $r(G) -x$. 
\end{itemize}
Thus, if $d_1$ is even at the end of step 2, then $k \leq r(G) - x$.\\
\textbf{Case(2):} Suppose $d_1$ is odd at the end of step 2.
\begin{itemize}
\item In step 2 of the algorithm, we take $j \leq \displaystyle \frac{(d_1-1)}{2}$. That is, we take vertices from all the $j$ neighborhoods of $v$, where $j \leq \displaystyle \frac{(d_1-1)}{2}$.
\item But after executing step 2 (at the end of step 2), we get 
\begin{equation*}
diam(\left\langle D_1 \right\rangle)  
\begin{cases}
=\left\lceil \displaystyle \frac{diam(G)}{l}\right\rceil-1, & \text{if $\left\langle D_1 \right\rangle$ is a tree}\\
\leq \left\lceil \displaystyle \frac{diam(G)}{l}\right\rceil-1, &\text{otherwise}. 
\end{cases}
\end{equation*}
\item So, as discussed in point 2 of Case(1), in order to make $diam(\left\langle D_1 \right\rangle) = \left\lceil \displaystyle \frac{diam(G)}{l}\right\rceil$, we insert some more vertices from the next neighborhoods of $v$ in to $D_1$ in step 6 of the algorithm. Hence $k$ becomes lesser than or equal to $r(G)-x$.
\end{itemize}
Thus, if $d_1$ is odd at the end of step 2, then $k \leq r(G)-x$.\\

 Thus, $k \leq r(G) - x$.
\end{proof} 
  
By the definition of $x$ and above theorem, we can write
\begin{equation*}
k \leq 
\begin{cases}
r(G) - \displaystyle \frac{diam(\left\langle D_1 \right\rangle)}{2}, & \text{if $diam(\left\langle D_1 \right\rangle)$ is even}\\ 
r(G) - \displaystyle \frac{diam(\left\langle D_1 \right\rangle)}{2}+\displaystyle \frac{1}{2}, & \text{otherwise.}
\end{cases}
\end{equation*}
Thus, 
\begin{equation}
\label{eq1}
k \leq r(G) - \displaystyle \frac{diam(\left\langle D_1\right\rangle)}{2} + \displaystyle \frac{1}{2}. 
\end{equation}

$l$- min HC team does not exist in any given network for  all values of $l$. So, we analyze the value of $l$ for which our algorithm finds a $l$-HC team in any given network, with the condition that $k \leq diam(\left\langle D_1 \right\rangle) = \left\lceil \displaystyle \frac{diam(G)}{l}\right\rceil$ (since in the algorithm $diam(\left\langle D_1 \right\rangle)$ is fixed as  $\left\lceil \displaystyle \frac{diam(G)}{l}\right\rceil$).\\

\subsection{Minimum Value of $l$ for non self centered networks}
First, we find out the minimum value of $l$ for a network, which is not self centered, such that $l$-HC team exists in the network.
In this subsection throughout, we do not consider the self centered graphs.
\begin{theorem}
If the given network $G$ is not self centered, then our algorithm always finds a $l$-HC team, for $l \leq \displaystyle \frac{3}{2}$.
\end{theorem}
\begin{proof}
Let $diam(\left\langle D_1 \right\rangle) = \left\lceil \displaystyle \frac{diam(G)}{l}\right\rceil$.\\
Let us write $r(G) = diam(G) - a$, where $1 \leq a \leq \displaystyle \frac{diam(G)}{2}$. As $G$ is not self centered, $a \neq 0$.\\
 $ diam(\left\langle D_1 \right\rangle) = \left\lceil \displaystyle \frac{diam(G)}{l}\right\rceil \geq \displaystyle \frac{diam(G)}{l} \Rightarrow [-diam(\left\langle D_1 \right\rangle)] \leq \displaystyle \frac{-diam(G)}{l}$.

Then by equation~\ref{eq1}, 

\begin{equation*}
\begin{split} 
k & \leq r(G) -\displaystyle \frac{diam(\left\langle D_1\right\rangle)}{2} +\displaystyle \frac{1}{2} \\
 & \leq [diam(G)-a]-\displaystyle \frac{diam(G)}{2l}+\displaystyle \frac{1}{2}. 
 \end{split}
\end{equation*}
By condition, $k\leq diam(\left\langle D_1 \right\rangle) = \left\lceil \displaystyle \frac{diam(G)}{l}\right\rceil \leq (\displaystyle \frac{diam(G)}{l})+1$ (since $\left\lceil y\right\rceil  \leq y + 1$).\\
$\Rightarrow [diam(G)-a]-\displaystyle \frac{diam(G)}{2l}+\displaystyle \frac{1}{2} \leq \displaystyle \frac{diam(G)}{l}+1$.\\
$\Rightarrow (2l-3)diam(G) \leq l(2a+1)$.\\
If $(2l-3) < 0$, then $[-diam(G)] \leq \displaystyle \frac{l(2a+1)}{(2l-3)}$.\\ 
$\Rightarrow diam(G) \geq \displaystyle \frac{-l(2a+1)}{ (2l-3)}$. \\
Let us write $ y = \displaystyle \frac{l(2a+1)}{(2l-3)}$. As $a\geq 1$, $\displaystyle \frac{l(2a+1)}{(2l-3)} \geq 1 \Rightarrow y \geq 1$.\\
Thus, $diam(G) \geq -y$, where $y$ is a positive quantity ($y \geq 1$).\\
This implies that if $(2l-3) < 0$, then $k\leq diam(\left\langle D_1 \right\rangle)$ is true for any graph $G$ with  $diam(G) \geq -y$, where $y \geq 1$ and hence it is true for any  
graph $G$ with  $diam(G) \geq 1$ (since  $diam(G)$ is always positive).\\
This implies that $k\leq diam(\left\langle D_1 \right\rangle)$ is true for any given network (not self centered) if $(2l-3) < 0$, that is, if $l < \displaystyle \frac{3}{2}$.\\
Thus, our algorithm finds a $l$-HC team in a given network (not self centered), if $l < \displaystyle \frac{3}{2}$.\\
Next let us consider  $l =\displaystyle \frac{3}{2}$.\\
$\Rightarrow diam(\left\langle D_1 \right\rangle) = \left\lceil \displaystyle \frac{2}{3} diam(G)\right\rceil$.\\
\textbf{Case(1):} Suppose $G$ is bi-eccentric. \\
$\Rightarrow r(G) =diam(G)-1$.\\ 
Then by equation~\ref{eq1}, 
\begin{equation*}
\begin{split}
k & \leq diam(G) - 1 - \displaystyle \frac{diam(G)}{3} +\displaystyle \frac{1}{2} \\
& = \displaystyle \frac{2}{3} diam(G) - \displaystyle \frac{1}{2} \\
& < \displaystyle \frac{2}{3} diam(G)\\
& = diam(\left\langle D_1 \right\rangle). 
 \end{split}
 \end{equation*}
$\Rightarrow k <  diam(\left\langle D_1 \right\rangle)$.\\
Thus, if $G$ is bi-eccentric, then our algorithm  finds a $\displaystyle \frac{3}{2}$-HC team.\\
For the  other graphs, let us write 
 $r(G) = diam(G) - a$, for $2 \leq a \leq \displaystyle \frac{diam(G)}{2}$.
Then by equation~\ref{eq1},
\begin{equation*}
\begin{split}
k & \leq diam(G) - a - \displaystyle \frac{diam(G)}{3} +\displaystyle \frac{1}{2} \\
& = \displaystyle \frac{2}{3} diam(G) - (a-\displaystyle \frac{1}{2}) \\
& < \displaystyle \frac{2}{3} diam(G),\  for \  a \geq 2.\\
& = diam(\left\langle D_1 \right\rangle).  
 \end{split}
 \end{equation*}
 $\Rightarrow k <  diam(\left\langle D_1 \right\rangle)$.\\
Thus our algorithm  finds a $\displaystyle \frac{3}{2}$-HC team in any given network (not self centered).\\ 
Thus, from the above discussions, it is clear that our algorithm  finds a $l$-HC team, for $l \leq \displaystyle \frac{3}{2}$ in any given network which is not self centered. 
\end{proof}
We have theoretically proved that $\displaystyle \frac{3}{2}$-HC team exists in any given network. But what is the greatest lower bound of $l$? Let us give an example for  $l=1.6$ such that 1.6-HC team does not exist in the network.\\
Let $G$ be a bi-eccentric network $\Rightarrow r(G) = diam(G)-1$. Let $diam(G) =50$.\\
 $\Rightarrow diam(\left\langle D_1 \right\rangle) = \left\lceil \displaystyle \frac{100}{1.6}\right\rceil = 32$.\\
 $\Rightarrow x = (32/2) = 16$.\\
 $\Rightarrow k \leq r(G) - x = 49 - 16 = 33 > diam(\left\langle D_1 \right\rangle)$.\\
Thus, 1.6-HC team does not exist in any given network. It may exist in some networks and may not exist in some networks. This shows that the greatest lower bound for $l$ is $\displaystyle \frac{3}{2}$, such that $l$-HC team exists in any given network.

\subsection{Minimum Value of $l$ for Self centered Networks}
For networks which are not self centered, the minimum value of $l$ is $\displaystyle \frac{3}{2}$. So, for self centered networks also, we start from $\displaystyle \frac{3}{2}$.\\
Suppose $G$ is self centered. $\Rightarrow r(G) =diam(G)$.\\
Let $l =\displaystyle \frac{3}{2}$.\\
$\Rightarrow diam(\left\langle D_1 \right\rangle) = \left\lceil \displaystyle \frac{2}{3} diam(G)\right\rceil$.\\
Then by equation~\ref{eq1}, 
\begin{equation*}
\label{eq:3}
\begin{split}
k & \leq diam(G) - \displaystyle \frac{diam(G)}{3} + \displaystyle \frac{1}{2} \\
& = \displaystyle \frac{2}{3} diam(G) + \displaystyle \frac{1}{2} \\
& > \displaystyle \frac{2}{3} diam(G) \\
& = diam(\left\langle D_1 \right\rangle).
\end{split}
\end{equation*}
$\Rightarrow k >  diam(\left\langle D_1 \right\rangle)$.\\
Theoretically it seems that   our algorithm does not find a $\displaystyle \frac{3}{2}$-HC team in a self centered network.

We tried  direct substitutions and found that $k \leq diam(\left\langle D_1\right\rangle)$ for  $l=\displaystyle \frac{3}{2}$ and hence our algorithm finds a $\displaystyle \frac{3}{2}$-HC team in self centered networks also.   Let us recall
\begin{equation*}
x = 
\begin{cases}
\displaystyle \frac{diam(\left\langle D_1 \right\rangle)}{2}, & \text{if $diam(\left\langle D_1 \right\rangle)$ is even}\\ 
\displaystyle \frac{[diam(\left\langle D_1 \right\rangle)-1]}{2}, & \text{otherwise.}
\end{cases}
\end{equation*}
Some examples are shown in the table ~\ref{tab:1}.\\
\begin{table}[h]
\caption{Direct Substitution for Self centered Graphs with $l=\displaystyle \frac{3}{2}$}
\label{tab:1}      
\begin{tabular}{llll}
\hline\noalign{\smallskip}
 $diam(G)$ &	 $diam(\left\langle D_1\right\rangle) = \left\lceil \displaystyle \frac{diam(G)}{1.5}\right\rceil$ & $x$ 		& $k = r(G) -x$ \\
\noalign{\smallskip}\hline\noalign{\smallskip} 
500& 334& 167& $333 < diam(\left\langle D_1\right\rangle)$\\
100 & 67& 33& $67 = diam(\left\langle D_1\right\rangle)$\\
99& 66 & 33 & $66 = diam(\left\langle D_1\right\rangle)$\\
81& 54 & 27 & $54 = diam(\left\langle D_1\right\rangle)$\\
50 & 34 & 17 & $33 < diam(\left\langle D_1\right\rangle)$\\
34& 23 & 11 & $23 = diam(\left\langle D_1\right\rangle)$\\
23& 16 & 8 & $15 < diam(\left\langle D_1\right\rangle)$\\
20& 14& 7& $13 < diam(\left\langle D_1\right\rangle)$\\

\noalign{\smallskip}\hline
\end{tabular}
\end{table}
\\
From the above two subsections, we infer that our algorithm finds a $l$-HC team in any given network, for $l \leq \displaystyle \frac{3}{2}$.
\begin{note}
The answer from direct substitution method and the answer from theoretical method will not differ much. For example, theoretically we proved that $k >  diam(\left\langle D_1 \right\rangle)$ for $l=1.6$ in self centered graphs. Now we try by direct substitution method. \\
Let $G$ be a self centered network and let $diam(G) =300$.\\
 $\Rightarrow diam(\left\langle D_1 \right\rangle) = \left\lceil \displaystyle \frac{300}{1.6}\right\rceil = 188$.\\
 $\Rightarrow x = (188/2) = 94$.\\
 $\Rightarrow k \leq r(G) - x = 300 - 94 = 206 > diam(\left\langle D_1 \right\rangle)$.\\
 Thus, by direct substitution also, we get the same result.\\ 
 So we have a question: Why do the answer from theoretical method and the answer from the direct substitution method differ for $l=\displaystyle \frac{3}{2}=1.5$ in  self centered graphs?
 
For $l=\displaystyle \frac{3}{2}$ in self centered graphs,  we obtained $k \leq \displaystyle \frac{2}{3} diam(G) + \displaystyle \frac{1}{2}$ theoretically (refer equation in the starting of this sub section). This quantity differs from $\displaystyle \frac{2}{3} diam(G)$ only by 0.5. If we take $\left\lfloor k\right\rfloor$, we get $ k = \displaystyle \frac{2}{3}diam(G)$. So, if the difference between $k$ and $diam(\left\langle D_1 \right\rangle)$ is less than 1, then we have to go for direct substitution method.
 \end{note}
 \begin{note}
 Suppose $G$ is a 2-self centered network and let $l=\displaystyle \frac{3}{2}$.\\
   $\Rightarrow diam(\left\langle D_1 \right\rangle) = \left\lceil \displaystyle \frac{2}{3}\ diam(G)\right\rceil = 2 =diam(G)$\\
   But, by definition, we need to find a network $\left\langle D \right\rangle$ whose diameter is not same as the diameter of the original network $G$.\\
   Thus, in 2-self centered networks, we should not take $l=\displaystyle \frac{3}{2}$. \\
But the only possibility for $l$ in 2-self centered networks is $l=2$, because if $l=2$, then $diam(\left\langle D_1 \right\rangle) = 1 < diam(G)$.
This implies that the team is a clique. So, dominating clique (if exists) forms a HC team in 2-self centered networks. If there is no dominating clique in $G$, then $G$ does not possess a $l$-HC team. 

As the given graph itself is having less diameter, there is no need for $l$-HC team in 2-self centered graphs or in any graphs whose diameter is 2.  Comfortability is mainly important  for large diameter graphs to reduce the dispersiveness in those graphs.  
 \end{note}
 \begin{note}
  $l$-min HC team exists in any given network, only if $l \leq\displaystyle \frac{3}{2}$. If $l > \displaystyle \frac{3}{2}$, then $l$-min HC team may or may not exist, because $l$-min HC team does not exist in any given network for all values of $l$. For example, in any cycle $C_n$, 2-min HC team does not exist, whereas $\displaystyle \frac{3}{2}$-min HC team always exists.
  
   Our algorithm finds $l$-HC team in any given network for $l \leq \displaystyle \frac{3}{2}$ and for $l > \displaystyle \frac{3}{2}$, our algorithm may or may not find a $l$-HC team. Anyhow we need to find a highly comfortable team. Our algorithm does it for a reasonable $l$. But the main problem lies in how much it is comfortable or in other way, how much dispersiveness is reduced. 
\end{note}

Next, we fix $l=2$ and analyze the conditions for which our algorithm finds a 2-HC team and the conditions for which it does not find a 2-HC team. 
\subsection{2-Highly Comfortable Team}
 As $l =2$,
\begin{equation}
\label{eq2}
diam(\left\langle D_1 \right\rangle) = \left\lceil \displaystyle \frac{diam(G)}{2} \right\rceil.
\end{equation} 
\begin{equation*}
\Rightarrow diam(\left\langle D_1 \right\rangle) = 
\begin{cases}
\displaystyle \frac{diam(G)}{2}, & \text{if $diam(G)$ is even}\\ 
\displaystyle \frac{(diam(G)+1)}{2}, & \text{otherwise.}
\end{cases}
\end{equation*}
First we discuss for even diameter of $G$.\\
\textbf{Case(1):} Suppose $diam(G)$ is even.\\
$\Rightarrow diam(\left\langle D_1 \right\rangle) = \displaystyle \frac{diam(G)}{2} \  \Rightarrow x = \displaystyle \frac{diam(G)}{4}$.\\
\textbf{Sub case(1):} Suppose $diam(G)=2r(G)$.\\
$\Rightarrow r(G) = \displaystyle \frac{diam(G)}{2}$.\\
$\Rightarrow k \leq \displaystyle \frac{diam(G)}{2}-(\displaystyle \frac{diam(G)}{4} -\displaystyle \frac{1}{2}) $, by equation~\ref{eq1}.\\
$\Rightarrow k \leq \displaystyle \frac{diam(G)}{4} + \displaystyle \frac{1}{2}$.\\
But $(\displaystyle \frac{diam(G)}{4} + \displaystyle \frac{1}{2}) \leq \displaystyle \frac{diam(G)}{2}$   if $diam(G) \geq 2$.\\
Thus, $k\leq diam(\left\langle D_1 \right\rangle)$ for all graphs of diameter greater than or equal to 2, if $r(G) = \displaystyle \frac{diam(G)}{2}$.\\
\textbf{Sub case(2):} Suppose  $diam(G)=2r(G) -2$.\\
$\Rightarrow r(G) = \displaystyle \frac{(diam(G)+2)}{2}$.\\
$\Rightarrow k \leq \displaystyle \frac{(diam(G)+2)}{2}-(\displaystyle \frac{diam(G)}{4} -\displaystyle \frac{1}{2}) = \displaystyle \frac{diam(G)}{4} + \displaystyle \frac{3}{2} $, by equation~\ref{eq1}.\\
But $\displaystyle \frac{diam(G)}{4} + \displaystyle \frac{3}{2} \leq \displaystyle \frac{diam(G)}{2}$   if $diam(G) \geq 6$.\\
Thus, $k\leq diam(\left\langle D_1 \right\rangle)$ for all graphs of diameter greater than or equal to 6, if $r(G) = \displaystyle \frac{(diam(G)+2)}{2}$.\\
 Let us generalize the above two sub cases as follows:\\
 Suppose  $diam(G)=2r(G) - b$, where $b$ is even.\\
$\Rightarrow r(G) = \displaystyle \frac{(diam(G)+b)}{2}$.\\
$\Rightarrow k \leq \displaystyle \frac{(diam(G)+b)}{2}- (\displaystyle \frac{diam(G)}{4} -\displaystyle \frac{1}{2}) = \displaystyle \frac{diam(G)}{4} + \displaystyle \frac{(b+1)}{2} $, by equation~\ref{eq1}.\\
But $[\displaystyle \frac{diam(G)}{4} + \displaystyle \frac{b}{2}+\displaystyle \frac{1}{2}] \leq \displaystyle \frac{diam(G)}{2}$   if $diam(G) \geq 2b+2$.\\
Thus, $k\leq diam(\left\langle D_1 \right\rangle)$ for all graphs of diameter greater than or equal to $(2b+2)$, if $r(G) = \displaystyle \frac{(diam(G)+b)}{2}$, where $b$ is even.\\
\textbf{Case(2):} Suppose $diam(G)$ is odd.\\
$\Rightarrow diam(\left\langle D_1 \right\rangle) = \displaystyle \frac{(diam(G)+1)}{2}$.\\
\textbf{Sub case(1):} Suppose  $diam(G)=2r(G) -1$.\\
$\Rightarrow r(G) = \displaystyle \frac{(diam(G)+1)}{2}$.\\
$\Rightarrow k \leq \displaystyle \frac{(diam(G)+1)}{2}-\displaystyle \frac{(diam(G)+1)}{4} +\displaystyle \frac{1}{2}= \displaystyle \frac{diam(G)}{4} + \displaystyle \frac{3}{4}$, by equation~\ref{eq1}.\\
But $\displaystyle \frac{diam(G)}{4} + \displaystyle \frac{3}{4} \leq \displaystyle \frac{(diam(G)+1)}{2}$   if $diam(G) \geq 1$.\\
Thus, $k\leq diam(\left\langle D_1 \right\rangle)$ for all graphs of diameter greater than or equal to 1, if $r(G) = \displaystyle \frac{(diam(G)+1)}{2}$.\\
In general,  
Suppose  $diam(G)=2r(G) - b$, where $b$ is odd.\\
$\Rightarrow r(G) = \displaystyle \frac{(diam(G)+b)}{2}$.\\
$\Rightarrow k \leq \displaystyle \frac{(diam(G)+b)}{2}- \displaystyle \frac{(diam(G)+1)}{4} +\displaystyle \frac{1}{2} = \displaystyle \frac{diam(G)}{4} + \displaystyle \frac{b}{2}+ \displaystyle \frac{1}{4} $, by equation~\ref{eq1}.\\
But $(\displaystyle \frac{diam(G)}{4} + \displaystyle \frac{b}{2}+ \displaystyle \frac{1}{4}) \leq \displaystyle \frac{(diam(G)+1)}{2}$   if $diam(G) \geq 2b-1$.\\
Thus, $k\leq diam(\left\langle D_1 \right\rangle)$ for all graphs of diameter greater than or equal to $(2b-1)$, if $r(G) = \displaystyle \frac{(diam(G)+b)}{2}$, where $b$ is odd.\\

Next, let us find the nature of $b$.\\
\textbf{Nature of $b$:}\\
Usually $0 \leq b \leq r(G)$. If $b = 0$, then $diam(G) = 2r(G)$ and if $b= r(G)$, then $G$ is self centered.\\
\textbf{Case(1):} Suppose $G$ is self centered.\\
 $\Rightarrow r(G) = diam(G)$.\\
\begin{equation*}
\begin{split}
 k & \leq diam(G) - (\displaystyle \frac{diam(G)}{4} -\displaystyle \frac{1}{2}) \\
 & =  \displaystyle \frac{3}{4} diam(G) + \displaystyle \frac{1}{2} \\
 & > \displaystyle \frac{diam(G)}{2} \\
 & = diam(\left\langle D_1 \right\rangle).
 \end{split}
 \end{equation*}
Thus, if $l = 2$, then our algorithm yields  $k > diam(\left\langle D_1 \right\rangle)$ for self centered networks. This implies that if $G$ is self centered, then our algorithm does not find a  2-highly comfortable team.\\
\textbf{Case(2):} Suppose $G$ is bi-eccentric.\\
$\Rightarrow r(G) = diam(G) -1 $.\\
\begin{equation*}
\begin{split}
 k & \leq [diam(G)-1] - \displaystyle \frac{diam(G)}{4} +\displaystyle \frac{1}{2} \\
 & = \displaystyle \frac{3}{4} diam(G) -\displaystyle \frac{1}{2} \\
 & > \displaystyle \frac{diam(G)}{2}, \ for \ diam(G) \geq 3\\
 & = diam(\left\langle D_1 \right\rangle).
 \end{split}
 \end{equation*}
Thus, if $l = 2$, then our algorithm yields  $k > diam(\left\langle D_1 \right\rangle)$, for $diam(G) \geq 3$. This implies that if $G$ is bi-eccentric, then our algorithm does not find a  2-highly comfortable team for $diam(G) \geq 3$. 

Proceeding like this, we infer that if \textbf{$b$ is a constant independent of  $r(G)$}, then our algorithm  finds a 2-highly comfortable team. If $b$ is dependent on $r(G)$, then our algorithm does not yield a 2-Highly comfortable team. 

Thus, 2-HC team exists for infinitely many cases. 2-HC team is important because dispersiveness is reduced $50\%$ only if $l=2$. 

Similarly, we can see that $l$-HC team exists for infinitely many cases and does not exist for some cases,  if $l \geq 3$. If the value of $l$ is increased, then dispersiveness is reduced much. But there are some disadvantages in increasing the value of $l$. We discuss in the next section.

\section{Reduction Rate of Dispersiveness}
\label{rate}
\begin{itemize}
	\item Our algorithm finds a $l$-HC team in any given network, for $l \leq \displaystyle \frac{3}{2}$.
	\item Although $l$ can be any value lesser than or equal to $\displaystyle \frac{3}{2}$,  
	we should always take the greatest value of $l$, because  if $l$ tends to 1, then by definition, $diam(D)$ tends to $diam(G)$, which again represents a dispersive team. 
	\item If $l$ is less, then dispersiveness in the team is more and hence the comfortability inside the team is less. 
	\item So, let us always take the maximum possible value of $l$. That is, we take 
	
\begin{enumerate}
	\item $l =2$ (if 2-HC team exists), 
	\item If 2-HC team does not exist, then we take $ \displaystyle \frac{3}{2} \leq l < 2$ (whichever $l$ is possible).
	\item As $\displaystyle \frac{3}{2}$ is the lower bound for $l$, we can take $l = \displaystyle \frac{3}{2}$ for any given network. 
	\end{enumerate} 
	\item This implies that our algorithm finds a HC team $\left\langle D \right\rangle$ from a network, whose diameter $diam(\left\langle D \right\rangle)$ is less than or equal to either $50 \%$ of $diam(G)$  ($l = 2$) or at least $66 \%$ of $diam(G)$ ($l =\displaystyle \frac{3}{2}$).
 \item This implies that our algorithm finds a HC team from a given network  whose dispersiveness is reduced at least 34$\%$ ($l \geq 1.5$) in all cases and whose dispersiveness is reduced to half ($l = 2$) in some cases (but infinitely many cases).
\end{itemize}
From the above points, we infer that our algorithm has reasonably reduced the dispersiveness of a team.

The above points pose the following question:
Why $l$ is not chosen above 2?
In fact, if $l$ increases, then dispersiveness is reduced much. For example, if $l= 3$, then dispersiveness is reduced $70\%$, if $l= 4$, then dispersiveness is reduced $75\%$, and so on.
But there are some disadvantages in increasing $l$. If $l$ is increased, then
\begin{enumerate}
	\item the number of members in the team is reduced ($|D|$ is reduced).
	\item the number of cases (graphs) for which $l$-HC team exists, also get reduced (that is, probability of existence of $l$-HC team is reduced).
	\item $k$ may become greater than $diam(\left\langle D_1 \right\rangle)$.
\end{enumerate}
Until now, we are discussing about only minimizing $D$ and hence $|D|$ can be reduced. So, $l$ can be increased. But in the section~\ref{max}, we consider the problem of maximizing $D$. In such a case, $l$ should not be increased. 

Also, as the algorithm is only an approximation algorithm, we can suitably take any value of $l$, such that  $k \leq diam(\left\langle D_1 \right\rangle)$ is not affected and 
which is giving better result. 
\section{Performance Ratio Of the Algorithm}
\label{ratio}
It is to be noted that the algorithm has three parameters, namely, $|D|$, $k^{*}$ and $l$. The set $D$ represents the team members of $l$-HC team, $k^{*}$ represents the dispersion index for the team and $l$ determines the existence of a HC team with $D$ and $k^{*}$. It is necessary that $D$, $k^{*}$ and $l$ should be minimized simultaneously. Also, $D$, $k^{*}$ and $l$ are inter related. 

In this section,  we give performance ratio for finding both $D$ and $k^{*}$, keeping $l$ fixed. 
Also, we prove some theorems and corollary, which give the  relations between the three parameters.

The following theorem gives the performance ratio for finding the $l$-highly reduced dispersive set $D$.

\begin{theorem}
\label{per}
The performance ratio of the algorithm for finding $l$-HC team  is at most $O(\ln \Delta(G))$.
\end{theorem}
\begin{proof}
Let $D$ be a minimum $l$-highly reduced dispersive $k^{*}-$ distance dominating set of $G$.\\
$\Rightarrow D$  is a connected $k^{*}-$ distance dominating set of $G$ (since any $l$-highly reduced dispersive set is a connected set).\\
$\Rightarrow D$  is a connected dominating set of $G^{k^{*}}$ (by definition of the graph $G^{k^{*}}$). \\
Performance ratio for finding $\gamma_c(G)$ is at most $O(\ln \Delta(G))$.\\
$\Rightarrow$ Performance ratio for finding $\gamma_c(G^{k^{*}})$ is at most $O(\ln \Delta(G^{k^{*}}))$.\\
But, $\Delta(G^{k^{*}}) \leq (\Delta(G)) ^ {k^{*}}$.\\
$\Rightarrow$ Performance ratio for finding $\gamma_c(G^{k^{*}})$ is at most $O(\ln (\Delta(G)) ^ {k^{*}}) = O(k^{*} \ln \Delta(G)) = O(\ln \Delta(G))$.\\
$\Rightarrow$ Performance ratio for finding the set $D$ is at most $O(\ln \Delta(G))$.\\
Thus, the performance ratio of the algorithm for finding $l$-HC team  is at most $O(\ln \Delta(G))$.
\end{proof}

Next, let us give the performance ratio for finding $k^{*}$. Before that, we prove a theorem relating $D$, $k^{*}$ and $l$, which helps us in finding the performance ratio of $k^{*}$.
\subsection{Properties of $k^{*}$}
\begin{theorem}
\label{big}
 $diam(G) \leq \displaystyle \frac{l(2k^{*} +1)}{(l-1)}$ for any network $G$ containing a $l$-HC team.
\end{theorem} 
\begin{proof}
Let $G$ be the given graph.\\
Let $D$ be a $l$-highly reduced dispersive set of $G$.\\
 By definition of $l$-HC team, $diam(\left\langle D\right\rangle) \leq \left\lceil \displaystyle \frac{diam(G)}{l} \right\rceil$ and \\ $k^{*} \leq diam(\left\langle D\right\rangle) \leq \left\lceil \displaystyle \frac{diam(G)}{l} \right\rceil$. \\
$\Rightarrow \left\langle D\right\rangle$ can be any sub graph of $G$ with $diam(\left\langle D\right\rangle) \leq \left\lceil \displaystyle \frac{diam(G)}{l} \right\rceil$.\\
Let us construct the graph $G$ assuming that $D$ and $k^{*}$ are given and  $\left\langle D\right\rangle$ is its $l$-HC team with  $k^{*}$ satisfying the above condition.\\
\textbf{Case(1):}\\
Assume that $k^{*} =1$.\\ 
First, let us construct $G$ assuming that $\left\langle D\right\rangle$ is a path, say $Q$, of length\\ $(\left\lceil \displaystyle \frac{diam(G)}{l} \right\rceil+1)$. 
$\Rightarrow diam(Q) = \left\lceil \displaystyle \frac{diam(G)}{l} \right\rceil.$
That is, let us construct $G$ assuming that the path $\left\langle Q\right\rangle$ is its $l$-HC team with  $k^{*} =1$.\\
As   $k^{*} =1$, it means that every vertex in the set $V-Q$ is reachable from the set $Q$ by a distance of one. So, we can add at least one vertex to each  vertex of the path $Q$. \\
\textbf{Sub case(1):}
Let us add at least one (pendant) vertex to \textbf{each of the vertices of the path $Q$}.\\ Now we can see that the newly constructed graph $G$ is nothing but a tree, whose diameter is increased by \textbf{two} from the diameter of the path $Q$. For example, consider the path $Q$ as in  Figure~\ref{fig5}. The set $V-Q$ represents the newly added vertices and the newly constructed $G$ = $ Q \cup (V-Q)$, which is a tree.\\ 
$\Rightarrow diam(constructed \  G) = diam(\left\langle Q\right\rangle) +2$.\\
 \begin{figure}
\footnotesize\centering
\centerline{\includegraphics[width=4in]{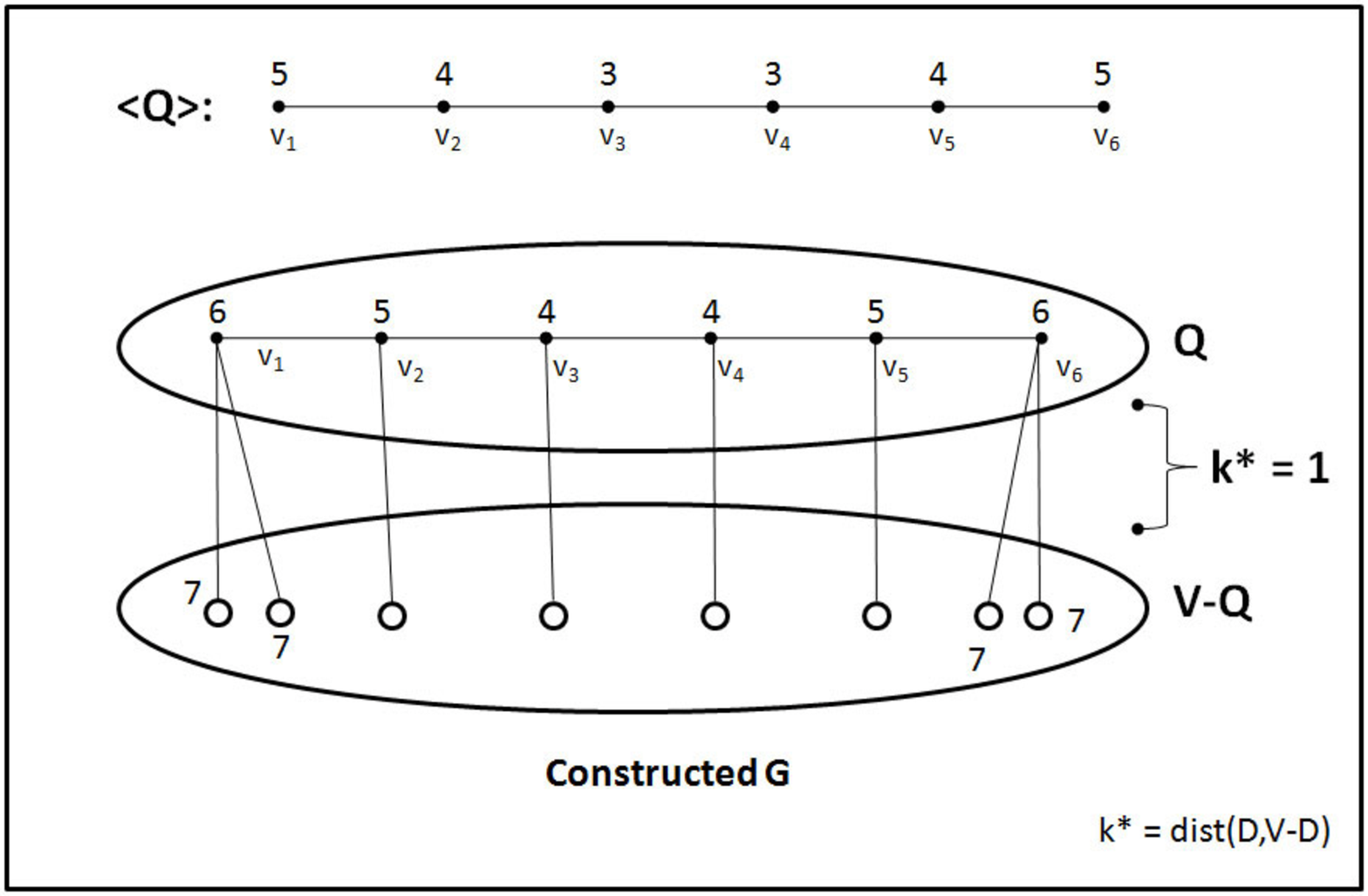}}
\caption{Construction of $G$ from a Path  with Diameter of the Path increased by Two}
\label{fig5}
\end{figure}
\textbf{Sub case(2):}
Let us add at least one vertex to each vertex of the path $Q$ \textbf{except exactly one peripheral vertex of  $Q$}.\\ Now we can see that the newly constructed graph $G$ is nothing but a tree, whose diameter is increased by \textbf{one} from the diameter of the path $Q$. For example, refer Figure~\ref{fig6}.\\
 $\Rightarrow diam(constructed \  G) = diam(\left\langle Q\right\rangle) +1$.
\begin{figure}
\footnotesize\centering
\centerline{\includegraphics[width=4in]{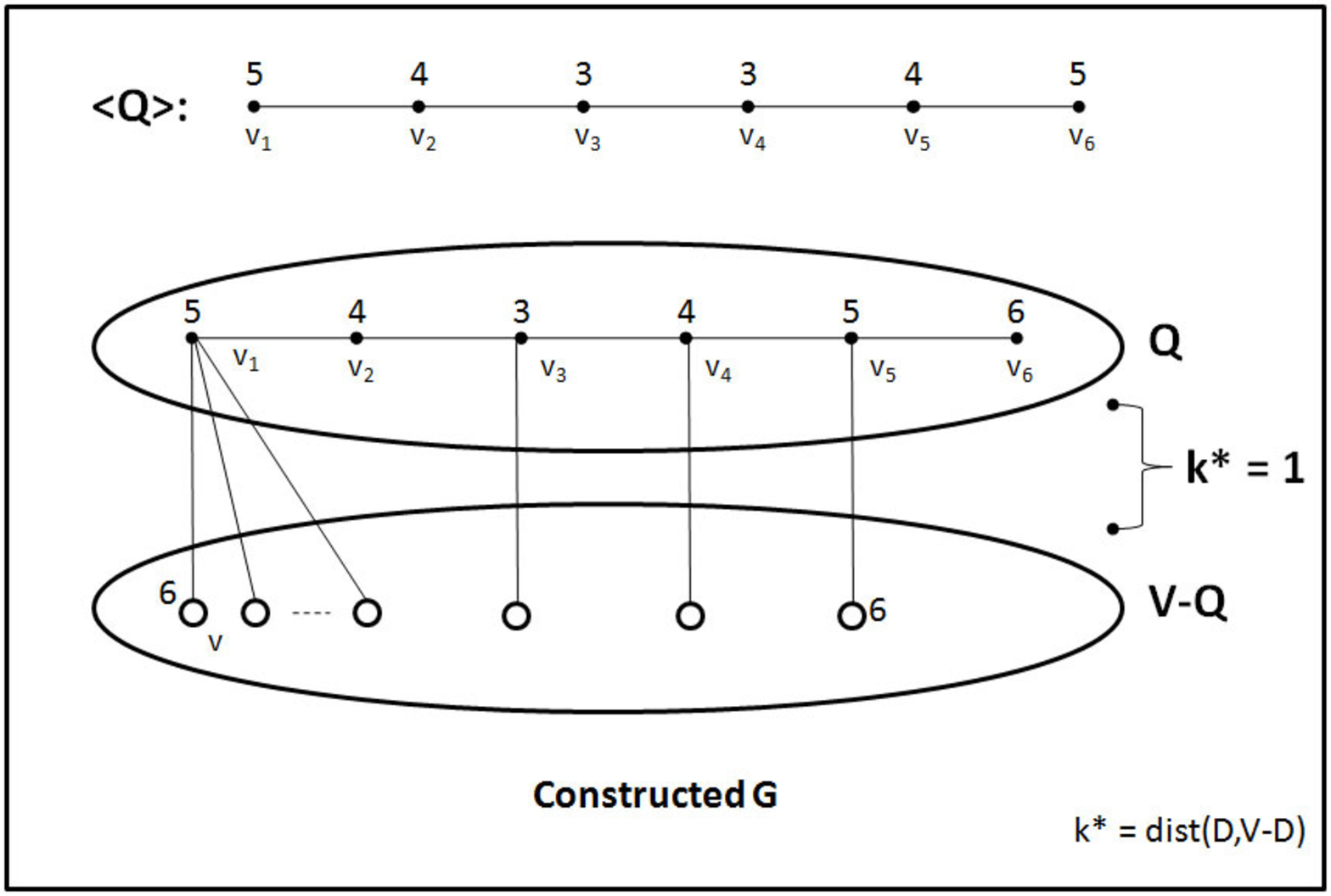}}
\caption{Construction of $G$ from a Path  with Diameter of the Path increased by One}
\label{fig6}
\end{figure}

Thus, from the Sub cases (1) and (2), we can infer that \\
$diam(constructed \  G) \leq diam(\left\langle Q\right\rangle) + 2$.\\
\textbf{Claim 1:} $diam(constructed \  G) \neq diam(Q) + c$, for $c \geq 3$.\\
That is, we claim that the diameter of the newly constructed $G$ \textbf{can not increase more than two} from the diameter of the path $Q$.\\
\textit{Proof of Claim 1:} As $k^{*} = 1$, every vertex in $V-Q$ should reach $Q$ at a distance exactly equal to 1 (not more than 1).\\
This implies that we \textbf{can not add a path of length more than 1} to each vertex of the path $Q$, that is we can not attach paths of length 2,3, etc. We can add only one vertex, which is nothing but a path of length 1.\\
This implies that the  diameter of new $G$ can not increase more than two from the diameter of the path $Q$. (because as discussed in the Sub cases (1) and (2), adding one vertex to a peripheral vertex of $Q$ will increase the diameter of $Q$ by one and adding one vertex to each of the peripheral vertices of $Q$ will increase the diameter of $Q$ by two).\\ Thus, $diam(constructed \  G) \neq diam(Q) + c$, for $c \geq 3$\\ and hence $diam(constructed \  G) \leq diam(\left\langle Q\right\rangle) + 2$.\\

Until now, we have proved that if $k^{*} = 1$ and $\left\langle D\right\rangle$ is a \textbf{path} of diameter $\left\lceil \displaystyle \frac{diam(G)}{l} \right\rceil$, then  
\begin{equation}
\label{eq:5}
diam(constructed \  G) \leq diam(\left\langle D \right\rangle) + 2.
\end{equation}

Next, let us prove that equation~\ref{eq:5} is true  for any sub graph $\left\langle D\right\rangle$ (not only paths).\\
Construction process of new $G$ is same as the above one, that is, in the worst case, we add at least one pendant vertex to each vertex of the sub graph $\left\langle D\right\rangle$.\\
$\Rightarrow diam(constructed \  G) = diam(\left\langle D \right\rangle) + 2$.\\ 
\textbf{Claim 2:} For any sub graph $\left\langle D \right\rangle$, $diam(constructed \  G) \neq diam(\left\langle D \right\rangle) + c$, for $c \geq 3$.\\
\textit{Proof for Claim 2:} It is to be noted that 
\begin{itemize}
	\item adding at least one (pendant) vertex to \textbf{one peripheral vertex} of \textbf{any sub graph} $\left\langle D \right\rangle$ will increase the diameter of $D$ by \textbf{one} and 
	\item adding at least one pendant vertex to \textbf{two peripheral vertices}, which are eccentric points of each other, will increase the diameter of $D$ by \textbf{two}. For example, refer Figure~\ref{fig7}.
\end{itemize}
\begin{figure}
\footnotesize\centering
\centerline{\includegraphics[width=4in]{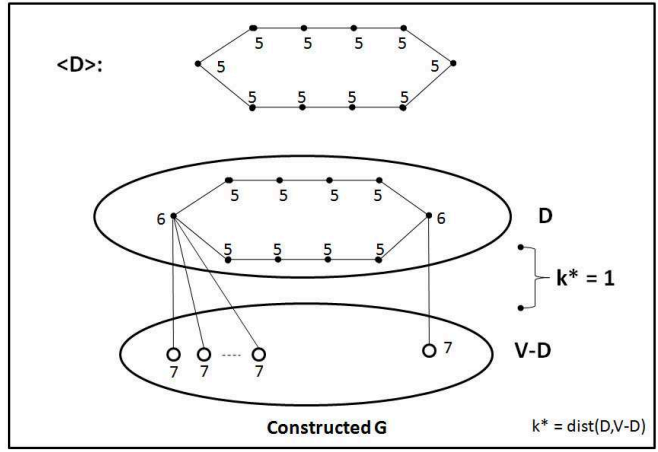}}
\caption{Construction of $G$ from General $D$ with Diameter of $D$ increased by Two}
\label{fig7}
\end{figure}

Also as discussed above, we can not add  a path of length more than one  to any vertex of $D$, because $k^{*} =1$.\\
This implies that diameter of the new $G$ can not be increased more than two from the diameter of $\left\langle D \right\rangle$.\\
Thus, for any sub graph $\left\langle D \right\rangle$, if $k^{*} = 1$, then \\ $diam(constructed \  G) \neq diam(\left\langle D \right\rangle) + c$, for $c \geq 3$.\\
Combining the above two claims, we get 
\begin{equation}
\label{eq:6}
diam(G) \leq diam(\left\langle D \right\rangle) + 2,\  if \ k^{*} = 1
\end{equation}
\textbf{Case(2):} Assume that $k^{*} = 2$.\\
This implies that every vertex in $V-D$ is reachable from $D$ by a distance of two. So, we can attach at least one  \textbf{path of length two} ($P_2$) to each vertex of the path $Q$ or any subgraph $\left\langle D \right\rangle$. For example, refer Figure~\ref{fig8}. As discussed in the Case (1),  we get $diam(constructed \  G) \leq diam(\left\langle D \right\rangle) + 4$. Also, as $k^{*} =2$, we \textbf{can not attach a path of length more than 2} to each vertex of $D$ and hence $diam(constructed \  G) \neq diam(\left\langle D \right\rangle) + c$, for $c \geq 5$.
\begin{figure}
\footnotesize\centering
\centerline{\includegraphics[width=4in]{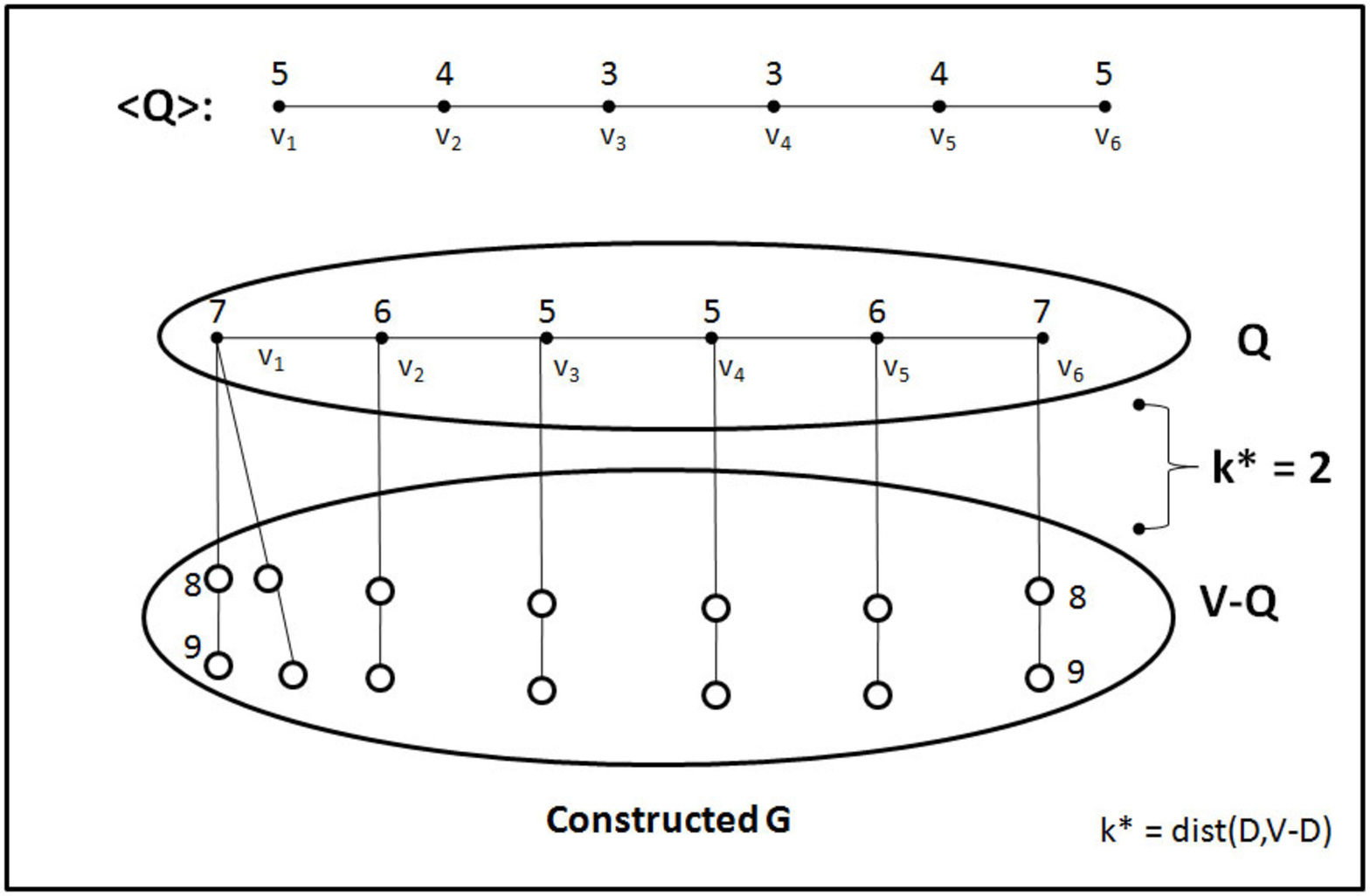}}
\caption{Construction of $G$ from a Path  with Diameter of Path increased by Four}
\label{fig8}
\end{figure}
Thus, 
\begin{equation}
\label{eq:7}
diam(G) \leq diam(\left\langle D \right\rangle) + 4,\  if \ k^{*} = 2
\end{equation}

Similar to the  discussions in the cases (1) and (2), we can construct a new $G$ for any $k^{*}$. Similar to the equations~\ref{eq:6} and \ref{eq:7}, we  get \\
$diam(G) \leq diam(\left\langle D \right\rangle) + 2k^{*}$.\\
 $\Rightarrow diam(G) \leq 2k^{*} + \left\lceil \displaystyle \frac{diam(G)}{l}\right\rceil$.\\
$\Rightarrow diam(G) \leq 2k^{*} +(\displaystyle \frac{diam(G)}{l}) + 1$ (since $\left\lceil x \right\rceil \leq x+1$).\\
$\Rightarrow diam(G) \leq \displaystyle \frac{l(2k^{*} +1)}{(l-1)}$.
\end{proof}
Next, we give the performance ratio of finding $k^{*}$ using the Theorem~\ref{big}.
\begin{theorem}
\label{PR}
The performance ratio for finding the dispersion index $k^{*}$ of a $l$-HC team is constant and $k^{*} \rightarrow \displaystyle \frac{2}{(l-1)}$, as $diam(G) \rightarrow \infty$.
\end{theorem}
\begin{proof}
From the Theorem~\ref{big}, we get, $k^{*} \geq [\displaystyle \frac{(l-1)}{l}diam(G)-1]/2$.\\
Also, by definition of $l$-HC team, $k \leq \left\lceil \displaystyle \frac{diam(G)}{l}\right\rceil \leq \displaystyle \frac{diam(G)}{l}+1$.\\
\begin{equation*}
\begin{split}
\Rightarrow \displaystyle\frac{k}{k^{*}} & \leq \displaystyle \frac{2(diam(G)+l)}{[(l-1)diam(G)-l]}\\
& \rightarrow \displaystyle \frac{2}{(l-1)},\ when\  diam(G) \rightarrow \infty.
\end{split}
\end{equation*}
\end{proof}
As we have $l=\displaystyle \frac{3}{2}$ as the lower bound, performance ration for finding $k^{*}$ in a $\displaystyle \frac{3}{2}$-HC team is 4. 
\begin{corollary}
\label{cor}
$l(k^{*}-1) \leq diam(G) \leq \displaystyle \frac{l(2k^{*} +1)}{(l-1)}$.
\end{corollary}
\begin{proof}
By definition of $l$-HC team, $k^{* } \leq diam(\left\langle D \right\rangle) \leq \left\lceil \displaystyle \frac{diam(G)}{l}\right\rceil$.\\
$\Rightarrow k^{* } \leq  \left\lceil \displaystyle \frac{diam(G)}{l}\right\rceil$.\\
$\Rightarrow k^{* } \leq  \left\lceil \displaystyle \frac{diam(G)}{l}\right\rceil + 1$ (since $\left\lceil x \right\rceil \leq x+1$).\\
$\Rightarrow l(k^{*}-1) \leq diam(G)$. Also, the upper bound follows from the Theorem ~\ref{big}.
\end{proof}
\begin{note}
It is to be noted that our algorithm can be applied to find $l$-HC team in \textbf{random networks} (graphs) also. From the theorems~\ref{per}, \ref{PR} and Corollary~\ref{cor}, it is clear that the performance ratio of the algorithm for finding $D$ and $k^{*}$ is dependent on  $diam(G)$ and $\Delta(G)$. As both these terms can be expressed in terms of the probability $p$, the performance ratio of the algorithm can be easily obtained for random networks in terms of $p$.

Also, if the network (graph) is disconnected, then as mentioned in the Section~\ref{intro}, algorithm can can be applied to each connected component of the network and hence \textbf{$l$-HC team can be obtained in disconnected networks also}. Thus, algorithm can be applied to find $l$-HC team in any given network. 
\end{note}
\section{HC team with maximum members}
\label{max}

In the above sections, we have always minimized $D$ . Note that $D$ can also be maximized. But $k^{*}$ should always be minimum irrespective of whether $D$ is minimized or maximized.

Let us state the problem as follows:
Maximize $D$ such that
\begin{enumerate}
\item $e_{\left\langle D\right\rangle}(v) <  e_G(v)$, for every vertex $v \in D$ (less dispersive)  
	\item $diam(\left\langle D\right\rangle) \leq \left\lceil \displaystyle \frac{diam(G)}{l}\right\rceil$ (measure of dispersion within $D$, that is, good communication among team members)  
	\item $dist(D, V-D) \leq k^{*}$ ($k^{*}-$ distance domination, that is, good service providers to the non team members)
	and \item $k^{*} \leq diam (\left\langle D\right\rangle)$ (easily accessible from the non team members).
\end{enumerate}
Maximal cardinality of a  minimum \lq $l$-highly reduced dispersive\rq \ $k^{*}-$ distance dominating set of $G$ is denoted by $\Gamma_{lhcomf}(G)$. Let us denote the HC team with maximum members as $l$-max HC team and HC team with minimum members as $l$- min HC team.\\
\textbf{Example 6:} Consider the graph $G$ as in Figure~\ref{fig9}.
\begin{figure}[h]
\footnotesize\centering
\centerline{\includegraphics[width=4in]{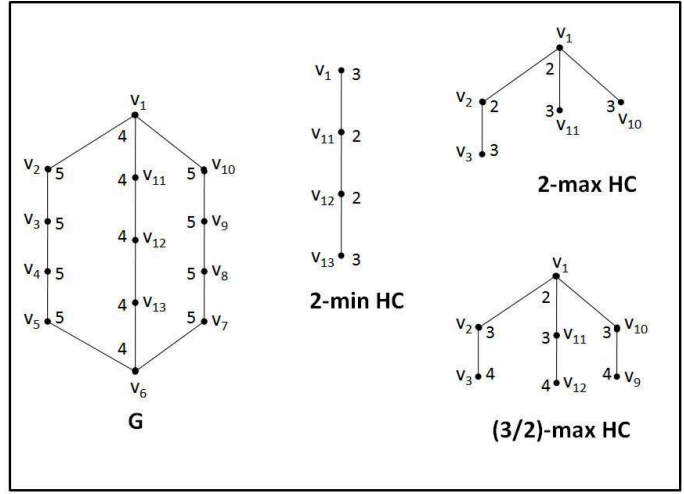}}
\caption{A Network, its min-HC Team and its max-HC Team}
\label{fig9}
\end{figure}
 First, let us take $l=2$. From the Figure~\ref{fig9}, we can see that $D_1 =\{v_1, v_{11}, v_{12}, v_{13}\}$ forms 2-highly reduced dispersive, 3-distance dominating set, because, $diam(\left\langle D_1\right\rangle) = 3= \left\lceil \displaystyle \frac{diam(G)}{2}\right\rceil$ and $k^{*} = 3 = diam(\left\langle D_1\right\rangle)$. 

Also, from the Figure~\ref{fig9}, we can see that $D_2 = \{v_1, v_2, v_3, v_{10}, v_{11}\}$ also forms 2-highly reduced dispersive, 3-distance dominating set, because, $diam(\left\langle D_1\right\rangle) = 3= \left\lceil \displaystyle \frac{diam(G)}{2}\right\rceil$ and $k^{*} = 3 = diam(\left\langle D_1\right\rangle)$.

Thus, for $l=2$ and $k^{*} = 3$, we get two 2-HC teams, one set with lesser number of vertices than the other set. This implies that $\left\langle D_1\right\rangle$ is the 2-min HC team and  hence $\gamma_{2hcomf}(G) = 4$. Also, $\left\langle D_2\right\rangle$ is the 2-max HC team  and $\Gamma_{2hcomf}(G) = 5$.

Next, let us fix $l =\displaystyle \frac{3}{2}$. We get the $\displaystyle \frac{3}{2}$-max HC team as in the Figure~\ref{fig9}.

\subsection{Advantages of $l$-max HC and $l$-min HC teams}
\textbf{Advantage of $l$-max HC team:} From definition, it is clear that $l$- max HC team contains more members than $l$-min HC team. So, if members are more, 
\begin{itemize}
	\item the time taken to complete a task will be fast  
	\item the average work done by a person in the team will be overall reduced.
	\item As the average work done by a person is less, stress for a person is less.
	\item As there are more members, more ideas will be shared  and so on.
\end{itemize}
 So, \textbf{more is the team power, fast is the work done and less is the stress for a person}.\\
\textbf{Advantage of $l$-min HC team:} As $l$-min HC team contains less members comparing to $l$-max HC team,  the  cost spent to maintain the  team is less, that is, maintenance cost is less. 

Our algorithm can be used to find both $l$-max HC team and $l$-min HC team. Let us discuss it as follows:
As discussed in the Section~\ref{rate},  \textbf{if $l$ is more, then people in the team are less}. So, in order to find $l$-min HC team, we can increase the value of $l$ and find $D_1$ from the algorithm such that it satisfies $k \leq diam(\left\langle D_1 \right\rangle)$. Similarly, in order to find $l$-max HC team, keep $l$  minimum and find $D_1$ from the algorithm such that it satisfies $k \leq diam(\left\langle D_1 \right\rangle)$. But if $l$ is reduced, the reduction rate of dispersiveness will be minimum. So, we must choose $l$ such that dispersiveness is reduced reasonably. As $l=\displaystyle \frac{3}{2}$ is the least value of $l$, which is giving good reduction in dispersiveness, we can always choose $l=\displaystyle \frac{3}{2}$ to find $l$-max HC team.
\section{Conclusion}
\label{conc}
In this paper, a new index called comfortability is defined in SNA. Based on this, three new definitions, namely comfortable team, better comfortable team and highly comfortable team  are given. It is proved that forming better comfortable team or highly comfortable  team in any given network are NP-complete. A polynomial time approximation algorithm is given for finding $l$-HC team in any given network and the time complexity of that algorithm is given. The various values of $l$ for which  $l$-HC team exists are analyzed and a lower bound for $l$ is obtained such that $l$-HC team exists.  From that analysis, it is proved that our algorithm has reasonably reduced the dispersion rate. Some of the structural properties are analyzed and using that performance ratio of the algorithm has been proved.  Highly comfortable team with maximum members, is defined. The advantages of minimum highly comfortable team and maximum highly comfortable team are discussed. It is also analyzed, how our algorithm finds both $l$-min HC team and $l$-max HC team for a reasonable value of $l$.
\subsection{Future Work}
The algorithm can be applied in a particular social network, for example, scale-free networks, and can be tried to reduce the performance ratio in that network. Algorithm can be applied to get exact values also in some particular networks. Further analysis can be made to find the performance ratio of the algorithm to find $l$-max HC team and to reduce the performance ratio for finding $l$-min HC team.



\end{document}